\documentclass[a4paper,noarxiv,unpublished,twocolumn]{quantumarticle}
\pdfoutput=1

\usepackage[utf8]{inputenc}
\usepackage[english]{babel}
\usepackage[T1]{fontenc}
\usepackage{amsmath}
\usepackage{hyperref}
\usepackage{lipsum}
\usepackage{multirow}

\usepackage{mathtools}
\usepackage{amssymb}
\usepackage{braket}
\usepackage{amsfonts}
\usepackage{mathabx}

\usepackage{algpseudocode}
\usepackage{graphicx}
\usepackage{color}
\usepackage{tikz}
\usetikzlibrary{arrows, shapes.geometric,shapes.symbols,shapes.multipart,
fit,positioning,shadows, chains, positioning, bending, shapes.arrows,
decorations.text, decorations.pathreplacing}
\usepackage{float}
\usepackage{dblfloatfix}
\renewcommand\floatpagefraction{0.9}

\usepackage{listings}

\usepackage{amsthm}
\newtheorem{theorem}{Theorem}[section]
\newtheorem{lemma}[theorem]{Lemma}
\newtheorem{observation}[theorem]{Observation}

\theoremstyle{definition}

\definecolor{gray}{RGB}{70,70,70}
\definecolor{white}{RGB}{255,255,255}
\definecolor{orange}{RGB}{255,158,61}
\definecolor{blue}{RGB}{19,158,251}

\definecolor{codegreen}{rgb}{0,0.6,0}
\definecolor{codegray}{rgb}{0.5,0.5,0.5}
\definecolor{codepurple}{rgb}{0.58,0,0.82}
\definecolor{backcolour}{rgb}{0.95,0.95,0.92}

\lstdefinestyle{mystyle}{
    backgroundcolor=\color{backcolour},
    commentstyle=\color{codegreen},
    keywordstyle=\color{magenta},
    numberstyle=\tiny\color{codegray},
    stringstyle=\color{codepurple},
    basicstyle=\ttfamily\footnotesize,
    breakatwhitespace=false,
    breaklines=true,
    captionpos=b,
    keepspaces=true,
    numbers=left,
    numbersep=5pt,
    showspaces=false,
    showstringspaces=false,
    showtabs=false,
    tabsize=4
}

\usepackage[numbers,sort&compress]{natbib}

\emergencystretch=\maxdimen
\graphicspath{{./images/}}

\usepackage{ulem}

\newcommand{\ts}{\textsuperscript}

\begin{document}

\title{Designing a Fast and Flexible Quantum State Simulator}
\author{Saveliy Yusufov}
\affiliation{Wells Fargo}
\author{Charlee Stefanski}
\orcid{0000-0001-9856-5955}
\affiliation{Wells Fargo}
\author{Constantin Gonciulea}
\orcid{0000-0001-5870-4586}
\affiliation{Wells Fargo}

\maketitle

\begin{abstract}
    This paper describes the design and implementation of Spinoza, a fast and flexible quantum simulator written in
    Rust.
    Spinoza simulates the evolution of a quantum system's state by applying quantum gates, with the core design
    principle being that a single-qubit gate applied to a target qubit preserves the probability of pairs of
    amplitudes corresponding to measurement outcomes that differ only in the target qubit.
    Multiple strategies are employed for selecting pairs of amplitudes, depending on the gate type and other
    parameters, to optimize performance.
    Specific optimizations are also implemented for certain gate types and target qubits.

    Spinoza is intended to enable the development of quantum computing solutions by offering developers a simple,
    flexible, and fast tool for classical simulation.
    In this paper we provide details about the design and usage examples.
    Furthermore, we compare Spinoza's performance against several other open-source simulators to demonstrate its
    strengths.
\end{abstract}

\section{Introduction}

Several types of users, including researchers, software developers, and students, use quantum simulators to experiment
with quantum algorithms and circuits in a controlled environment without needing access to actual quantum hardware.
Classical simulation is also an essential tool for analyzing results from quantum hardware.
Classical simulators are essential in the Noisy Intermediate-Scale Quantum (NISQ) era since currently available
quantum devices have limited capacity for depth, circuit complexity, and high error rates~\cite{Preskill2018}.
By using classical simulators, researchers can continue to develop and test quantum algorithms despite these
limitations.
However, simulating quantum computers with classical computers is a difficult task.
The computational resources and simulation time increases exponentially with the size of the quantum system (qubits).
This complexity also demonstrates the potential utility and promise of quantum computing to perform tasks beyond the
reach of classical computers.

Spinoza is designed to fulfill the need for a quantum simulator that is flexible and efficient for use in several
different environments, including personal computers.
Spinoza offers all the essential capabilities for quantum algorithm development with both Rust and Python
interfaces.
In standard benchmarks, Spinoza is one of the fastest simulators when compared with several other open-source simulators.

This paper consists of the following: In Section~\ref{sec:prelim} discusses the preliminary concepts imperative to
the design and implementation of Spinoza.
Section~\ref{sec:design} provides an overview of the design.
Section~\ref{sec:impl} details the implementation in Rust and in addition to the Python module.
Section~\ref{sec:opt} discusses the optimization techniques used in the simulator.
Section~\ref{sec:benchmarks} compares the performance of Spinoza with several existing simulators.

\begin{figure}[!ht]
    \centering
    \includegraphics[width=.5\linewidth]{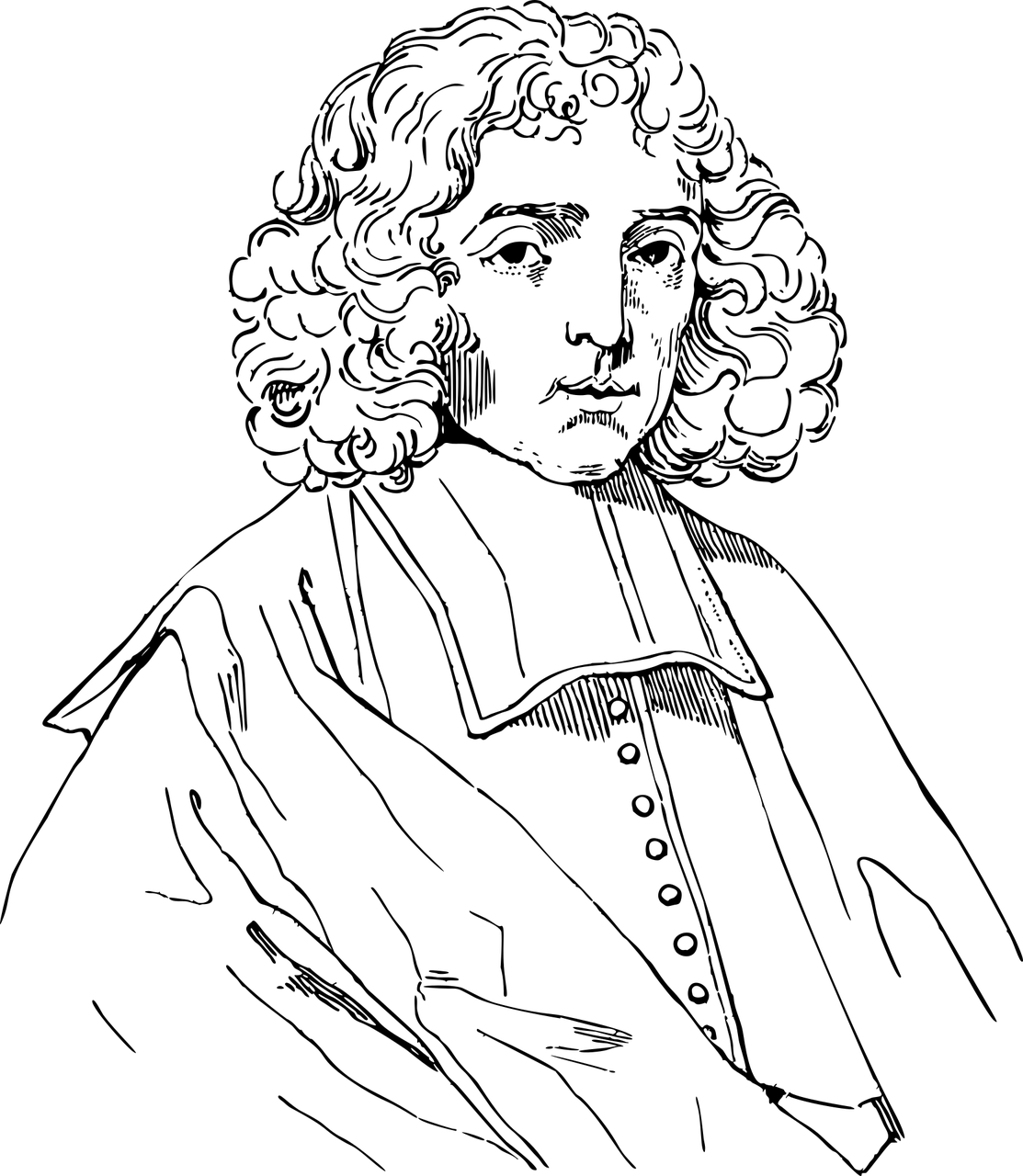}
    \captionof{figure}{\textbf{Baruch Spinoza (1632-1677)~\cite{spinoza}} This
    project is called ``Spinoza'', named after Baruch Spinoza, the 17\ts{th} century Dutch Philosopher.}
\end{figure}

\section{\label{sec:prelim}Preliminaries}

The state of a quantum system with $n > 0$ qubits can be modeled as a list of amplitudes corresponding to the
$2^{n}$ possible outcomes of the state measurement in the computational basis.
This is typically expressed using Dirac's ket notation:

\begin{equation}
    \label{eqn:gamma_state}
    \ket{\gamma}_n = \sum_{i = 0}^{2^n-1} a_{i} \ket{i}
\end{equation}

The probability of the outcome corresponding to $\ket{i}$ is $|a_{i}|^2$, for $i \in \{0, \mathellipsis, 2^{n} -1
\}$, and $\sum_{i = 0}^{2^{n}-1} |a_{i}|^2 = 1$.

Depending on the context, an outcome can be interpreted as an integer, $i \in \{0, \mathellipsis, 2^{n} -1 \}$, or
its binary string representation, where each digit corresponds to a qubit.
For example, a $3$-qubit quantum state can be expressed as:

\begin{equation}
    \label{eqn:gamma_3}
    \begin{split}
        \ket{\gamma}_{3} = \text{ } &a_{0}\ket{000} + a_{1}\ket{001} + a_{2}\ket{010} + a_{3}\ket{011} + \\
            &a_{4}\ket{100} + a_{5}\ket{101} + a_{6}\ket{110} + a_{7}\ket{111}.
    \end{split}
\end{equation}

\section{\label{sec:design}Approach and Design}

At the core of the design is the fact that single-qubit gate transformations act on amplitude pairs, preserving the
probability of individual pairs.
A unified approach for handling target and control qubits allows for an efficient implementation of multi-control
transformations.
Gate transformations are optimized according to gate type.

Spinoza's implementation of the evolution of the state of a quantum system by applying a (controlled) gate has two
steps:

\begin{enumerate}
    \item Select the pairs of amplitudes that are updated together. Note that
        for some gates (e.g., Phase gate) only one side of the pair is needed.
    \item Apply the gate to update either one or both sides of each such pair.
\end{enumerate}

When a single-qubit gate is applied on a target qubit, $t \in \{0, \ldots,
n-1\}$, the probability of the two outcomes that differ only in the $t$
position is preserved. We can rewrite the state defined in
Eq.~\ref{eqn:gamma_state} by grouping such pairs:

\begin{equation}
    \label{eqn:pairs}
    \ket{\gamma}_n = \sum_{j = 0}^{2^{n-1} - 1} \left( a_{z(j, t)} \ket{z(j, t)} +
    a_{o(j, t)} \ket{o(j, t)} \right)
\end{equation}

where $z(j, t)$ is the $n-1$ digit binary expansion of $j$ with $0$ inserted in
position $t$, and $o(j, t)$ is the $n-1$ digit binary expansion of $j$ with $1$
inserted in position $t$.

For example, given a $3$-qubit state, $\ket{\gamma}_{3}$ (as defined in
Eq~\ref{eqn:gamma_3}), a single-qubit gate applied to the middle qubit (i.e.,
$t = 1$) would yield the following pairs of outcomes:
\begin{equation*}
    \begin{split}
        & z(0,1) = 0 = 0\textbf{0}0_2 \text{ and } o(0, 1) = 2 = 0\textbf{1}0_2 \\
        & z(1,1) = 1 = 0\textbf{0}1_2 \text{ and } o(1, 1) = 3 = 0\textbf{1}1_2 \\
        & z(2,1) = 4 = 1\textbf{0}0_2 \text{ and } o(2, 1) = 6 = 1\textbf{1}0_2 \\
        & z(3,1) = 5 = 1\textbf{0}1_2 \text{ and } o(3, 1) = 7 = 1\textbf{1}1_2 \\
    \end{split}
\end{equation*}

Considering these outcome pairs, we can define the state in
Eq.~\ref{eqn:gamma_3} by pair groupings:

\begin{equation*}
    \begin{split}
        \ket{\gamma}_3 = \text{ } &(a_{0}\ket{000} + a_{2}\ket{010}) + \\
        &(a_{1}\ket{001} + a_{3}\ket{011}) + \\
        &(a_{4}\ket{100} + a_{6}\ket{110}) + \\
        &(a_{5}\ket{101} + a_{7}\ket{111})
    \end{split}
\end{equation*}

\subsection{Pair Selection}\label{subsec:pair-selection}

When a single-qubit gate is applied to a target qubit, it changes the amplitudes and corresponding probabilities of
measuring $0$ or $1$ for that qubit.
The effect of the gate application can be understood as preserving the probability of pairs of outcomes that differ
only in the target qubit, as formalized in Eq.~\ref{eqn:pairs}.

Four strategies for selecting pairs are described below.
For each strategy, we use the example of a $3$-qubit quantum state and the middle qubit as the target.
Namely, we set $t = 1$.

The list of binary expressions of the outcomes with $0$ in the target position is:

\begin{equation*}
    [0\textbf{0}0, \text{ } 0\textbf{0}1, \text{ } 1\textbf{0}0, \text{ } 1\textbf{0}1],
\end{equation*}

Similarly, the list of corresponding outcomes with $1$ in the target position is:

\begin{equation*}
    [0\textbf{1}0, \text{ } 0\textbf{1}1, \text{ } 1\textbf{1}0, \text{ } 1\textbf{1}1].
\end{equation*}

\subsubsection{Strategy 0: Traverse and Recognize}\label{subsubsec:strategy0}

We start with the list of all possible outcomes represented as binary strings:

\begin{equation*}
    [000, \text{ } 001, \text{ } 010, \text{ } 011, \text{ } 100, \text{ } 101, \text{ } 110, \text{ } 111].
\end{equation*}

For a system with $n > 0$ qubits, we traverse the list, and we check if an outcome $i \in \{0, \ldots, 2^{n} -1\}$
has $1$ in the target qubit position, $t \in \{0, \ldots, n - 1\}$. In code, the condition can be implemented as the
boolean expression \texttt{(i $\gg$ t) \& 1}.

We have also found three other strategies that offer better performance depending on the gate.
They have the following observation in common:

\begin{observation}\label{obs}
For a given prefix-suffix combination  $j \in \{0, \ldots, 2^{n-1}\}$, target
position $t \in \{0, \ldots, n-1\}$, if $q$ and $r$ are the quotient and
remainder, respectively, of dividing $j$ by $2^{t}$, (i.e., $j = q \cdot 2^{t}
+ r$, with $r \in \{0, \ldots, 2^{t} - 1\}$), then

\begin{equation}
    \begin{split}
        z(j, t) &= 2q \cdot 2^{t} + r \\
        o(j, t) &= (2q + 1) \cdot 2^{t} + r
    \end{split}\label{eq:q_and_r}
\end{equation}

The quotient $q$ and the remainder $r$ are the prefix and suffix with respect
to the target position $t$ in the binary representation of the outcome
corresponding to $j$.
\end{observation}

\subsubsection{Strategy 1: Group and Traverse}\label{subsubsec:strategy1}

We start with the list of all possible measurement outcomes represented as binary strings.
We split the list in chunks of size $2^{t}$.
In this example $t = 1$, so the chunk size is $2^{t} = 2$.
The chunks alternate between containing $2^{t}$ items with $0$ in the target position, and $2^{t}$ items with $1$ in the
target position.
Figure~\ref{fig:pair_shift_ex_t1} illustrates this strategy.

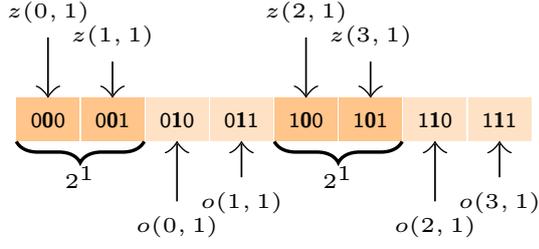
\begin{figure}[ht]
    \resizebox{\columnwidth}{!} {
        \begin{tikzpicture}[scale=0.5, every node/.style={scale=0.5}, font=\sffamily,every label/.append
        style={font=\tiny\sffamily,align=center}]
            % >>>>>>>>>>>>> NODES <<<<<<<<<<<<<<<<
            \tikzset{brace/.style={decoration={brace, mirror},decorate}}

            \node[rectangle,fill=orange!60, ultra thin, minimum height = 3.05mm, minimum width = 5mm,
                font=\tiny\sffamily](t1){0\textbf{0}0};
            \node[rectangle,fill=orange!60, ultra thin, minimum height = 3.05mm, minimum width = 5mm, font=\tiny\sffamily,
                right=0.05mm of t1](t2){0\textbf{0}1};
            \node[rectangle,fill=orange!30, ultra thin, minimum height = 3.05mm, minimum width = 5mm, font=\tiny\sffamily,
                right=0.05mm of t2](t3){0\textbf{1}0};
            \node[rectangle,fill=orange!30, ultra thin, minimum height = 3.05mm, minimum width = 5mm, font=\tiny\sffamily,
                right=0.05mm of t3](t4){0\textbf{1}1};
            \node[rectangle,fill=orange!60, ultra thin, minimum height = 3.05mm, minimum width = 5mm, font=\tiny\sffamily,
                right=0.05mm of t4](t5){1\textbf{0}0};
            \node[rectangle,fill=orange!60, ultra thin, minimum height = 3.05mm, minimum width = 5mm, font=\tiny\sffamily,
                right=0.05mm of t5](t6){1\textbf{0}1};
            \node[rectangle,fill=orange!30, ultra thin, minimum height = 3.05mm, minimum width = 5mm, font=\tiny\sffamily,
                right=0.05mm of t6](t7){1\textbf{1}0};
            \node[rectangle,fill=orange!30, ultra thin, minimum height = 3.05mm, minimum width = 5mm, font=\tiny\sffamily,
                right=0.05mm of t7](t8){1\textbf{1}1};

            \draw [brace,decoration={raise=0.01ex}] (t1.south west) -- node [font=\tiny\sffamily, yshift=-2ex]
                {$2^1$}(t2.south east);
            \draw [brace,decoration={raise=0.01ex}] (t5.south west) -- node [font=\tiny\sffamily, yshift=-2ex]
                {$2^1$}(t6.south east);

            \draw[<-, very thin](t1.north) -- +(0,5mm) node[above, font=\tiny\sffamily]{$z(0, 1)$};
            \draw[<-, very thin](t2.north) -- +(0,3mm) node[above, font=\tiny\sffamily]{$z(1, 1)$};
            \draw[<-, very thin](t5.north) -- +(0,5mm) node[above, font=\tiny\sffamily]{$z(2, 1)$};
            \draw[<-, very thin](t6.north) -- +(0,3mm) node[above, font=\tiny\sffamily]{$z(3, 1)$};

            \draw[<-, very thin](t3.south) -- +(0,-5mm) node[below, font=\tiny\sffamily]{$o(0, 1)$};
            \draw[<-, very thin](t4.south) -- +(0,-3mm) node[below, font=\tiny\sffamily]{$o(1, 1)$};
            \draw[<-, very thin](t7.south) -- +(0,-5mm) node[below, font=\tiny\sffamily]{$o(2, 1)$};
            \draw[<-, very thin](t8.south) -- +(0,-3mm) node[below, font=\tiny\sffamily]{$o(3, 1)$};

        \end{tikzpicture}
    }
    \captionof{figure}{Visualization of the group and traverse strategy of pair
    selection for a $3$-qubit state with $t=1$.}
    \label{fig:pair_shift_ex_t1}
\end{figure}

This strategy is an optimization of Strategy $0$, (see Section~\ref{subsubsec:strategy0}).
We leverage the fact that within a given chunk of length $2^{t}$, the target qubit does not change.
This strategy is especially useful when each side of the pair can be updated independent of the other side.
In the case of the $R_{z}$-gate, this strategy provides significant optimization opportunities, which are discussed in
Section~\ref{sec:opt}.

\subsubsection{Strategy 2: Concatenate Prefix, Target and Suffix}\label{subsubsec:strategy2}

Strategy 2 requires us to generate the prefix, append the target, and finally generate the suffix.
The prefix and suffix are the possible values expressed with the qubits before and after the target qubit,
respectively.
For this example, the possible prefix values are $0$ and $1$.

For each possible prefix, we append $0$ or $1$ in the target position:

\begin{equation*}
    \begin{split}
    & 00 \text{  } 01 \\
    & 10 \text{  } 11. \\
    \end{split}
\end{equation*}

The possible generated suffixes are $0$ and $1$. Each possible suffix is appended to each combination of prefix and
target:

\begin{equation*}
    \begin{split}
        & 000 \text{ } 010 \\
        & 001 \text{ } 011 \\
        & 100 \text{ } 110 \\
        & 101 \text{ } 111 \\
    \end{split}
\end{equation*}

With the same notations as in Observation~\ref{obs}, this strategy can use two
nested loops on $q$ and $r$ that generate the outcomes that have $0$ or $1$ in
the target position as in Eq.~\ref{eq:q_and_r}. This strategy is illustrated
in Figure~\ref{fig:strategy_2_diagram}.

\begin{figure}[ht]
    \resizebox{\columnwidth}{!} {
        \begin{tikzpicture}[scale=0.5, every node/.style={scale=0.5}, font=\sffamily,every label/.append
        style={font=\tiny\sffamily,align=center}]
            % >>>>>>>>>>>>> NODES <<<<<<<<<<<<<<<<

            % prefix box
            \node[rectangle, font=\scriptsize\sffamily, fill = white](p1){\textcolor{blue}{0}};
            \node[rectangle, font=\scriptsize\sffamily, fill = white, below=0.4mm of p1](p2){\textcolor{blue}{1}};
            \node[draw, dashed, thin, rounded corners, fit = (p1) (p2), inner xsep=8pt,inner ysep=12pt,
            label={above:{Prefix}}](fit1){};

            % target box
            \node[rectangle, font=\scriptsize\sffamily, fill = white, right=5mm of p1](t1){\textcolor{green}{0}};
            \node[rectangle, font=\scriptsize\sffamily, fill = white, below=0.4mm of t1](t2){\textcolor{green}{1}};
            \node[draw, dashed, thin, rounded corners, fit = (t1) (t2), inner xsep=8pt,inner ysep=12pt,
            label={above:{$t$}}](fit5){};

            % suffix box
            \node[rectangle, font=\scriptsize\sffamily, fill = white, right=5mm of t1](s1){\textcolor{teal}{0}};
            \node[rectangle, font=\scriptsize\sffamily, fill = white, below=0.4mm of s1](s2){\textcolor{teal}{1}};
            \node[draw, dashed, thin, rounded corners, fit = (s1) (s2), inner xsep=8pt,inner ysep=12pt,
            label={above:{Suffix}}](fit2){};

            % pairs
            \node[rectangle, font=\scriptsize\sffamily, fill = white, right=7mm of s1](pair1){
                (\textcolor{blue}{0}\textcolor{green}{0}\textcolor{teal}{0},
                \textcolor{blue}{0}\textcolor{green}{1}\textcolor{teal}{0})};

            \node[rectangle, font=\scriptsize\sffamily, fill = white, below=0.4mm of pair1](pair2){
                (\textcolor{blue}{0}\textcolor{green}{0}\textcolor{teal}{1},
                \textcolor{blue}{0}\textcolor{green}{1}\textcolor{teal}{1})};

            \node[rectangle, font=\scriptsize\sffamily, fill = white, below=0.4mm of pair2](pair3){
                (\textcolor{blue}{1}\textcolor{green}{0}\textcolor{teal}{0},
                \textcolor{blue}{1}\textcolor{green}{1}\textcolor{teal}{0})};

            \node[rectangle, font=\scriptsize\sffamily, fill = white, below=0.4mm of pair3](pair4){
                (\textcolor{blue}{1}\textcolor{green}{1}\textcolor{teal}{1},
                \textcolor{blue}{1}\textcolor{green}{1}\textcolor{teal}{1})};

            \node[rectangle, font=\tiny\sffamily, fill = white, above=0.1mm of pair1](pair_title){Pair};

            \draw[blue, ->, very thin] (p1.east) to (t1.west);
            \draw[blue, ->, very thin] (p1.east) to (t2.west);
            \draw[blue, ->, very thin] (p2.east) to (t1.west);
            \draw[blue, ->, very thin] (p2.east) to (t2.west);
            \draw[green, ->, very thin] (t1.east) to (s1.west);
            \draw[green, ->, very thin] (t1.east) to (s2.west);
            \draw[green, ->, very thin] (t2.east) to (s1.west);
            \draw[green, ->, very thin] (t2.east) to (s2.west);
            \draw[red!50, ->, very thin] (s1.east) to (pair1.west);
            \draw[purple!80, ->, very thin] (s2.east) to (pair2.west);
            \draw[gray!40!yellow!60, ->, very thin] (s1.east) to (pair3.west);
            \draw[gray!25!purple!60, ->, very thin] (s2.east) to (pair4.west);
        \end{tikzpicture}
    }
    \captionof{figure}{Example of generating pairs for a $3$-qubit state with
        target qubit set to $1$ (i.e., $t = 1$) for generating pairs using
        strategy 2; generating the possible prefixes, appending $0$ or $1$ in
        the target position, then generating the possible suffixes and
    appending each to each prefix-target
combination.}\label{fig:strategy_2_diagram}
\end{figure}
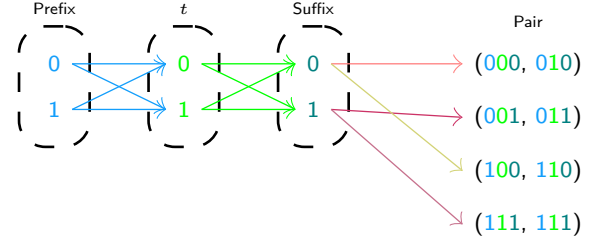

\subsubsection{Strategy 3: Insert Target}

Strategy 3 requires us to generate the prefix-suffix combinations, and then insert the target.
The possible prefix-suffix combinations are:

\begin{equation*}
    00 \text{ } 01 \text{ } 10 \text{ } 11.
\end{equation*}

To generate the pair items with $0$ in the target position, we insert $0$ into the target position, in this case the
middle:

\begin{equation*}
    000 \text{ } 001 \text{ } 100 \text{ } 101.
\end{equation*}

To generate the corresponding pair items with $1$ in the target position, we insert $1$ into the target position of
the same prefix-suffix combination:

\begin{equation*}
    010 \text{ } 011 \text{ } 110 \text{ } 111.
\end{equation*}

This strategy is illustrated in Figure~\ref{fig:strategy_3_diagram}.

\begin{figure}[ht]
    \resizebox{\columnwidth}{!} {
        \begin{tikzpicture}[scale=0.5, every node/.style={scale=0.5}, font=\sffamily,every label/.append
        style={font=\tiny\sffamily,align=center}]
            % >>>>>>>>>>>>> NODES <<<<<<<<<<<<<<<<
            \node[rectangle, font=\scriptsize\sffamily, fill = white](p1){0};
            \node[rectangle, font=\scriptsize\sffamily, fill = white, below=0.4mm of p1](p2){1};
            \node[draw, dashed, thin, rounded corners, fit = (p1) (p2), inner xsep=8pt,inner ysep=12pt,
            label={above:{Prefix}}](fit1){};

            \node[rectangle, font=\scriptsize\sffamily, fill = white, right=3mm of p1](s1){0};
            \node[rectangle, font=\scriptsize\sffamily, fill = white, below=0.4mm of s1](s2){1};
            \node[draw, dashed, thin, rounded corners, fit = (s1) (s2), inner xsep=8pt,inner ysep=12pt,
            label={above:{Suffix}}](fit2){};

            \node[rectangle, font=\scriptsize\sffamily, fill = white, above right=-0.5mm and 4mm of fit2](c1)
                {\textcolor{blue}{0}\textcolor{green}{\_} \textcolor{teal}{0}};

            \node[rectangle, font=\scriptsize\sffamily, fill = white, right=4mm of c1](pair1){
                (\textcolor{blue}{0}\textcolor{green}{0}\textcolor{teal}{0},
                \textcolor{blue}{0}\textcolor{green}{1}\textcolor{teal}{0})};

            \node[rectangle, font=\scriptsize\sffamily, fill = white, below=0.25mm of c1](c2){\textcolor{blue}{0}
            \textcolor{green}{\_} \textcolor{teal}{1}};

            \node[rectangle, font=\scriptsize\sffamily, fill = white, right=4mm of c2](pair2){
                (\textcolor{blue}{0}\textcolor{green}{0}\textcolor{teal}{1},
                \textcolor{blue}{0}\textcolor{green}{1}\textcolor{teal}{1})};

            \node[rectangle, font=\scriptsize\sffamily, fill = white, below=0.25mm of c2](c3){\textcolor{blue}{1}
            \textcolor{green}{\_} \textcolor{teal}{0}};

            \node[rectangle, font=\scriptsize\sffamily, fill = white, right=4mm of c3](pair3){
                (\textcolor{blue}{1}\textcolor{green}{0}\textcolor{teal}{0},
                \textcolor{blue}{1}\textcolor{green}{1}\textcolor{teal}{0})};

            \node[rectangle, font=\scriptsize\sffamily, fill = white, below=0.25mm of c3](c4){\textcolor{blue}{1}
            \textcolor{green}{\_} \textcolor{teal}{1}};

            \node[rectangle, font=\scriptsize\sffamily, fill = white, right=4mm of c4](pair4){
                (\textcolor{blue}{1}\textcolor{green}{1}\textcolor{teal}{1},
                \textcolor{blue}{1}\textcolor{green}{1}\textcolor{teal}{1})};

            \node[rectangle, font=\tiny\sffamily, fill = white, above=0.1mm of pair1](pair_title){Pair};

            \draw[red!50, ->, very thin] (p1.east) to (s1.west);
            \draw[red!50, ->, very thin] (s1.east) to (c1.west);
            \draw[red!50, ->, very thin] (c1.east) to (pair1.west);
            \draw[purple!80, ->, very thin] (p1.east) to (s2.west);
            \draw[purple!80, ->, very thin] (s2.east) to (c2.west);
            \draw[purple!80, ->, very thin] (c2.east) to (pair2.west);
            \draw[gray!40!yellow!60, ->, very thin] (p2.east) to (s1.west);
            \draw[gray!40!yellow!60, ->, very thin] (s1.east) to (c3.west);
            \draw[gray!40!yellow!60, ->, very thin] (c3.east) to (pair3.west);
            \draw[gray!25!purple!60, ->, very thin] (p2.east) to (s2.west);
            \draw[gray!25!purple!60, ->, very thin] (s2.east) to (c4.west);
            \draw[gray!25!purple!60, ->, very thin] (c4.east) to (pair4.west);

        \end{tikzpicture}
    }
    \captionof{figure}{Example of generating pairs for a 3-qubit state with
    target qubit $1$ (i.e., $t = 1$) for generating pairs using strategy 3---generating
prefix-suffix combinations, and then inserting $0$ or $1$ in the target
position.}\label{fig:strategy_3_diagram}
\end{figure}
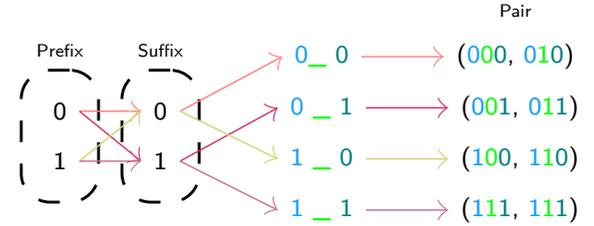

This strategy can be easily generalized for handling control qubits.

If $n > 0$ is the number of binary digits of the possible outcomes, then this strategy can be represented with the
formula below:

\begin{lemma}
    With the same notations as in Observation~\ref{obs}, the following
    closed-form expressions can be used for pair generation:
    \begin{equation*}
        \label{eqn:bit_shift}
        \begin{split}
            z(j, t) &= j + ((j \gg t) \ll t) \\
            o(j, t) &= j + (((j \gg t) + 1) \ll t)
        \end{split}
    \end{equation*}
\end{lemma}

\begin{proof}

    \begin{equation*}
        \begin{split}
            z(j, t) &= 2(q 2^{t}) + r \\
            &= q2^{t} + q2^{t} + r \\
            &= j + q 2^{t} \\
            &= j + (q \ll t) \\
            &= j + ((j \gg t) \ll t) \\
            \\
            o(j, t) &= 2^t + z(j, t) \\
            &= j + (((j \gg t) + 1) \ll t)
        \end{split}
    \end{equation*}

\end{proof}

Spinoza is implemented such that each gate uses the pair strategy that is optimized for the application of that gate.

\subsection{\label{subsec:amplitude-update}Amplitude Update}

The method of pairing outcomes and amplitudes is at the core of the simulator's implementation of single-qubit gate
applications.
Assume a single-qubit gate is represented in matrix form:

\begin{equation*}
    g = \begin{bmatrix} g_{00} & g_{01} \\ g_{10} & g_{11} \end{bmatrix}.
\end{equation*}

Let $U_{g}$ be the unitary operator corresponding to this single-qubit gate, then its application to a $n$-qubit
quantum state $\ket{\gamma}_{n}$, as defined in Eq.~\ref{eqn:gamma_state}, yields

\begin{equation*}
    U_g\ket{\gamma}_n = \sum_{j = 0}^{2^{n-1}-1} \left( b_{z(j, t)} \ket{z(j, t)}_n +  b_{o(j, t)} \ket{o(j, t)}_n
    \right)
\end{equation*}

where

\begin{equation*}
    \begin{bmatrix} b_{z(j, t)}  \\ b_{o(j, t)}  \end{bmatrix} =
    \begin{bmatrix} g_{00} & g_{01} \\ g_{10} & g_{11} \end{bmatrix}
    \begin{bmatrix} a_{z(j, t)}  \\ a_{o(j, t)}  \end{bmatrix}.
\end{equation*}

Therefore,

\begin{equation*}
    \label{eqn:transform_state}
    \begin{split}
        %    U_g\ket{a}   & = \sum_{j = 0}^{2^{n-1}-1} ((g_{00}a_{z(j, t)} + g_{01}a_{o(j, t)}) \ket{z(j, t)}_n\\
        %    & = \frac{1}{\sqrt{M}}\sum_{k=0}^{M-1}\left(\cos(k\theta) + i\sin(k\theta) \right) \ket{k}_m
        U_g\ket{\gamma}_n  & = \sum_{j = 0}^{2^{n-1}} ( (g_{00}a_{z(j, t)} + g_{01}a_{o(j, t)}) \ket{z(j, t)}_n \\
        &  + (g_{10}a_{z(j, t)} + g_{11}a_{o(j, t)}) \ket{o(j, t)}_n).
    \end{split}
\end{equation*}

Note that qubits are indexed from the right.
Also, to simplify understanding, sometimes we do not distinguish between a gate and its unitary.

\begin{figure}[ht]
    \resizebox{\columnwidth}{!} {
        \begin{tikzpicture}[scale=0.5, every node/.style={scale=0.5}, font=\sffamily,every label/.append
        style={font=\tiny\sffamily,align=center}]
            % >>>>>>>>>>>>> NODES <<<<<<<<<<<<<<<<
            % transformations
            \node[tape, draw, thin, tape bend=none,fill=blue!60!green!45, minimum width=2cm, double copy shadow,
                minimum height=2cm](gates) {TRANSFORMATIONS};

            % processor
            \node[rectangle, font=\small\sffamily, minimum width = 5cm, minimum height = 5cm, draw=blue!40, right=0.6cm
            of gates, label=above:{PROCESSOR}](worker){};
            \node[rectangle, font=\small\sffamily, minimum width = 4cm, minimum height = 1cm, fill=blue!60,
                above right=-0.75cm and -2.25cm of worker](pair_select){pair selection};
            \node[rectangle, font=\small\sffamily, minimum width = 4cm, minimum height = 2cm, fill=blue!20,
                right=-2.25cm of worker, below = 0.5cm of pair_select](update_amps){update amplitudes};

            \node[draw, rectangle split, rectangle split parts = 8, minimum height=1cm, label=above:{STATE}, rectangle
            split  horizontal, fill=orange!60, ultra thin, below=0.5cm of worker](state){$z_0$ \nodepart{two}$z_1$
                \nodepart{three}$z_2$ \nodepart{four}$z_3$  \nodepart{five}$\cdots$ \nodepart{six}$\cdots$
                \nodepart{seven}$\cdots$\nodepart{eight}$z_{N-1}$};

            % >>>>>>>>>>>>> CONNECTIONS <<<<<<<<<<<<<<<<
            \draw[->, > = stealth, dashed, thin] (gates.east) to (worker.west);
            \draw[->, > = stealth, dashed, thin] (pair_select.south) to node[left, font=\scriptsize\sffamily] {pair
            indices} (update_amps.north);
            \draw[->, > = stealth, thin] (update_amps.east) to [out=30,in=30](pair_select.east);
            \draw[->, > = stealth, thin] ([xshift=-7.5mm]update_amps.south) to (state.four north);
            \draw[->, > = stealth, thin] (state.seven north) to ([xshift=13mm]update_amps.south);

        \end{tikzpicture}
    }
    \captionof{figure}{Visualization of the serial implementation of the
    simulator. Transformations are passed to the single processor in order, and
the amplitudes of the state are updated in pairs.}\label{fig:diagram_generic}
\end{figure}
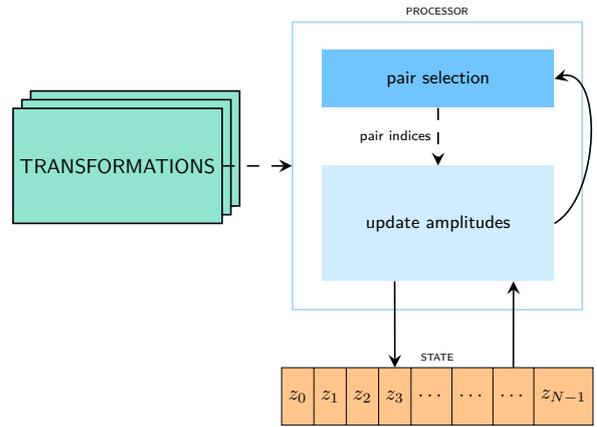

\subsection{\label{subsec:controlled-gate-transformations}Controlled Gate Transformations}

When one or more control qubits are part of a gate transformation, only the
amplitude pairs corresponding to outcomes that have $1$ in the control
positions are being processed.

By default, Spinoza uses the insertion strategy (i.e., Strategy 3), for controlled gate applications.

\section{\label{sec:impl}Implementation}

\subsection{Using Rust}

The exponential growth in run-time for simulating a quantum computer with a
classical machine prompts the need for an efficient simulator. Implementing a
fast quantum simulator requires the use of a programming language that provides
blazingly fast performance, efficient use of memory, and a low level of
control. The C and C++ programming languages fit the aforementioned criteria,
but with many potential issues that are beyond the scope of this paper. The
Rust programming language~\cite{rustbook} suffices all the aforementioned
requirements for a fast quantum simulator, and Rust brings additional benefits
to the table.

The primary benefit of Rust comes from the borrow checker, which helps
eliminate memory safety bugs that continue to plague projects written in C and
C++. Given the safety guarantees that are enforced by the compiler, project
maintainers can spend less time reviewing code submissions.

By implementing the simulator in Rust, quantum computing researchers and quantum software engineers can spend more
time adding new functionality and running simulations, and less time verifying correctness of their code.

\subsection{\label{subsec:functions}Code Structure}

\subsubsection{Single vs. Double Precision}

The implementation discussed and used for all benchmarks (Section~\ref{sec:benchmarks}) in this paper uses
double-precision floating-point values (i.e., variable type with $64$ bits) for representing the amplitudes of
quantum states and the values used for gate transformations.
Users should consider using single-precision (i.e., $32$ bit) floating-point types for additional performance
benefits.
Using single precision allows for more of the state vector to fit into the cache.
For the convenience of users, Spinoza can easily be built and installed to use single-precision complex numbers by
invoking the \texttt{single} feature flag.
In our testing, single-precision complex numbers offered a 30\% performance improvement.
For several use cases, double-precision values may not be required.
Note that the the use of single-precision complex numbers implies that the state vector only requires half the
memory as compared to the amount of memory required when using double-precision complex numbers.
Hence, the use of single-precision complex numbers gives users the ability to introduce an extra qubit into their
quantum state simulation, at the cost of precision.
For example, in the case of $32$ qubits, the state vector would require $2^{32} \cdot 128 \approx 68.72$ GB of memory
when using double-precision complex numbers.
On a machine with only $64$ GB of memory, this would lead to page faults, which would dramatically reduce the
performance of the simulation.
By using single-precision complex numbers, the state vector would only require $2^{32} \cdot 64 \approx 34.36$ GB of
memory, which would allow for fast simulation of the quantum state without reducing the number of qubits.

\subsubsection{\label{subsubsec:state_vec}Quantum State Representation}

A quantum state can be represented as a vector of complex numbers.
In Spinoza, the state is defined as a structure consisting of two vectors of single or double-precision
floating-point types---one for real components and one for the imaginary components of each amplitude---and an
unsigned one byte integer value $n$, or the number of qubits represented by the state.
The Rust implementation of the state vector data structure used in Spinoza  is shown in Appendix~\ref{sec:rust_code},
listing~\ref{statevector}.
Ostensibly, the state vector can be defined as a single vector of single or double precision complex numbers;
however, separating the real and imaginary components into two separate vectors creates opportunities for optimizations
that are not possible with a single vector.
Specific optimizations that leverage this attribute are discussed in Section~\ref{sec:simd}.
The same data structure is used for multiple registers of qubits in one quantum circuit.

There is a method for initializing a state vector, which takes an unsigned integer and initializes a state vector to
$\ket{0}$.
This function is used in the Spinoza example in listing~\ref{functional_example} in Appendix~\ref{sec:rust_examples}.

\subsubsection{Gate Structure}

In a general case, a single-qubit gate is defined as an array of four complex numbers.
Similar to the amplitudes of a quantum state, each complex number is defined by two single or double precision
floating-point types for the real and imaginary components.
While the convention is to define a gate with a two-by-two matrix, a one-dimensional representation offers better
memory locality, less memory consumption, and less allocations~\cite{so}.

For each gate, the most efficient strategy for identifying or generating the pairs is implemented as a method of the
gate.

\subsubsection{Pairs}

Rust implementations for each of the pair selection strategies described in Section~\ref{subsec:pair-selection} are
in Appendix~\ref{sec:rust_code}, listings~\ref{strat0} through~\ref{strat3}.
The best performing pair selection strategy for each gate type is used.
Furthermore, additional optimizations for pair selection are used in specific cases, depending on the gate type and
target qubit.
Note that in the current implementation of Spinoza, the traverse and recognize strategy (strategy 0) is not
used.

When applying single-qubits gates $Y$, $Z$, $P$, $R_{y}$, and $U$ the concatenation strategy (strategy 2) is used for
generating the pairs, as shown  in Appendix~\ref{sec:rust_code}, listing~\ref{strat2}.
The group and traverse strategy (strategy 1), as shown in Appendix~\ref{sec:rust_code}, listing~\ref{strat1} is used
for applying the $R_{z}$ gate.
When applying the single-qubit gates $X$, $H$, and $R_{x}$, the concatenation strategy (strategy 2) is also used, with
an additional optimization for when the target qubit is $0$, which is explained in Section~\ref{sec:opt}.

When applying control and multi-control gate transformations, pair selection strategies can be optimized.
For example, a modified version double shift strategy (strategy 3), as shown in Appendix~\ref{sec:rust_code},
listing~\ref{strat3_control}, is only used for controlled gate applications.

\subsubsection{Gate Transformation/Amplitude Update}

As discussed in Section~\ref{sec:design}, gate transformations are executed by identifying pairs of amplitudes and
updating those amplitudes according to the gate coefficients.
In Spinoza, gate transformations can be applied to a state using the functional syntax used in the example in
Appendix~\ref{sec:rust_code} listing~\ref{strat1}, or using the circuit syntax discussed in the next section.
The functional approach gives the developer control and flexibility when designing an algorithm.

\subsubsection{Quantum Circuits}

Spinoza includes methods that allow for quantum circuit development with a circuit syntax comparable to
Qiskit's~\cite{qiskit}.
Listing~\ref{circuit_example} in Appendix~\ref{sec:rust_examples} defines and executes a quantum circuit using the
circuit syntax.
Two data structures makeup the circuit components: registers and circuits.
A \texttt{QuantumRegister} is a vector of integers (\texttt{usize}) which correspond to indices of a state.
A \texttt{QuantumCircuit} is a vector of transformations and a state vector which reflects the register or registers
used in the circuit.
A quantum transformation consists of the gate type, target qubit, and optional controls and parameters.

To build a circuit, one or more registers of qubits is initialized, each with a given number of qubits
(Appendix~\ref{sec:rust_examples} listing~\ref{circuit_example}, line 2).
When initializing a circuit, one or more registers of qubits can be used, but only one state structure is defined in
the implementation.
The \texttt{QuantumRegister} facilitates indexing to the intended qubit when writing a circuit.
Gate transformations are then added to the circuit using methods of \texttt{QuantumCircuit} which include common
gate types (Pauli-gates, $U$-gate), controlled gates and multi-control gates.
In Appendix~\ref{sec:rust_examples} listing~\ref{circuit_example}, Hadamard ($H$) gates are added to the circuit on
line 6, and Phase ($P$) gates are added on line $10$.
For convenience when writing quantum circuits, Spinoza has an inverse quantum Fourier transform function,
\texttt{iqft}, that can be added to a circuit (or applied to a state) with the same syntax as a single gate
application (line $14$).
The transformations are not applied to the state until the \texttt{execute} method is called, as shown in
listing~\ref{circuit_example} line 15.
Using the circuit implementation allows for additional optimizations as the order and quantity of gates can be known
before execution.

\subsubsection{Measurement}

We provide a sampling method, \texttt{get\_samples}, that can be used to simulate measurement results for a given
number of shots.
Sampling is executed by utilizing weighted reservoir sampling~\cite{reservoir-sampling}.
An example of sampling (using the Python interface) is shown in Appendix~\ref{sec:python_examples}.

\subsection{\label{subsec:py_interface}Python Interface}

Python is a popular and frequently used programming language in the scientific community.
Spinoza can also be used as a Python library to facilitate easy installation and development for users.
The Spinoza Python library is designed to be similar to Qiskit to make it easy for researchers to use.
The library is implemented using bindings which call Rust code to perform the computations.
The Python bindings were implemented using PyO3~\cite{pyo3}.

Using Python to call Rust functions adds a negligible overhead.
The Python wrapper is used for the benchmarks in figure~\ref{fig:benchmarks}.

\section{\label{sec:opt}Optimizations}

\subsection{Gate Optimizations}

\subsubsection{X-Gate}

\begin{equation*}
    X =
    \begin{bmatrix}
    0 & 1 \\
    1 & 0
    \end{bmatrix}
\end{equation*}

Since the $X$-gate simply swaps the two sides of the pair, we can completely avoid any arithmetic operations.

When \texttt{target = 0}, we have a special case in which \texttt{distance = 1}.
In this case, we can iterate over the state by stepping by two elements per iteration.

\subsubsection{Y-Gate}

\begin{equation*}
    Y =
    \begin{bmatrix}
    0 & -i \\
    i & 0
    \end{bmatrix}
\end{equation*}

Let the two sides of the pair be $\gamma_{0}, \gamma_{1}$, where $\gamma_{0} = a + ib$ and $\gamma_{1} = c + id$.
The application of the $Y$-gate to this pair can be expressed as follows:

\begin{equation*}
\begin{split}
    Y\begin{bmatrix} \gamma_{0} \\ \gamma_{1} \end{bmatrix} &=
    \begin{bmatrix}
        0 & -i \\
        i & 0
    \end{bmatrix}
    \begin{bmatrix}
        \gamma_{0} \\
        \gamma_{1}
    \end{bmatrix} \\
    &=
    \begin{bmatrix}
        0 & -i \\
        i & 0
    \end{bmatrix}
    \begin{bmatrix}
        a + ib \\
        c + id
    \end{bmatrix} \\
    &=
    \begin{bmatrix}
         d - ic \\
        -b + ia
    \end{bmatrix}
\end{split}
\end{equation*}

\noindent
We can avoid the extra computations associated with matrix-vector multiplication by noticing that the pair can be
``updated'' by swapping the vector components, swapping the real and imaginary parts in each component, and by
flipping the signs.

\noindent
In addition, we apply a loop nest optimization (LNO) when using the concatenation strategy.
Namely, while appending the target, it is apparent that \texttt{((i >> target) << target)} rarely changes.
Rather, we can explicitly compute the prefixes.
As a result, the iteration logic transforms into the following: \texttt{base + 0, base + 1, base + 2, \ldots,
base + n}.
This transformation presents a vectorization opportunity for the compiler.
On the other hand, computing \texttt{i + ((i >> target) << target)} may cause a jump to something ``random'', so the
compiler is unable to leverage vectorization.

\subsubsection{Z-Gate}

\begin{equation*}
    Z =
    \begin{bmatrix}
        1 &  0 \\
        0 & -1
    \end{bmatrix}
\end{equation*}

Let the two sides of the pair be $\gamma_{0}, \gamma_{1}$, where $\gamma_{0} = a + ib$ and $\gamma_{1} = c + id$.
The application of the $Z$-gate to this pair can be expressed as follows:

\begin{equation*}
\begin{split}
    Z\begin{bmatrix} \gamma_{0} \\ \gamma_{1} \end{bmatrix} &=
    \begin{bmatrix}
        1 &  0 \\
        0 & -1
    \end{bmatrix}
    \begin{bmatrix}
        \gamma_{0} \\
        \gamma_{1}
    \end{bmatrix} \\
    &=
    \begin{bmatrix}
        1 &  0 \\
        0 & -1
    \end{bmatrix}
    \begin{bmatrix}
        a + ib \\
        c + id
    \end{bmatrix} \\
    &=
    \begin{bmatrix}
         a + ib \\
        -c - id
    \end{bmatrix}
\end{split}
\end{equation*}

\noindent
Given the above result, it is clear that only one side of the pair is needed when applying the $Z$-gate.
Traversing the state while applying the $Z$-gate will have greater spatial locality as compared to the application
of other gates.
Namely, regardless of the \texttt{target}, the application of the $Z$-gate will be a cache friendly operation.
Thus, we can reduce the number of operations associated with matrix-vector, as well as reduce the required memory
bandwidth.

\subsubsection{Hadamard Gate}

\begin{equation*}
    H =
    \begin{bmatrix}
        \frac{1}{\sqrt{2}} & \frac{1}{\sqrt{2}} \\
        \frac{1}{\sqrt{2}} & -\frac{1}{\sqrt{2}}
    \end{bmatrix}
\end{equation*}

Let the two sides of the pair be $\gamma_{0}, \gamma_{1}$, where $\gamma_{0} = a + ib$ and $\gamma_{1} = c + id$.
The application of the $H$-gate to this pair can be expressed as follows:

\begin{equation*}
\begin{split}
    H\begin{bmatrix} \gamma_{0} \\ \gamma_{1} \end{bmatrix} &=
    \begin{bmatrix}
        \frac{1}{\sqrt{2}} & \frac{1}{\sqrt{2}} \\
        \frac{1}{\sqrt{2}} & -\frac{1}{\sqrt{2}}
    \end{bmatrix}
    \begin{bmatrix}
        \gamma_{0} \\
        \gamma_{1}
    \end{bmatrix} \\
    &=
    \begin{bmatrix}
        \frac{1}{\sqrt{2}} & \frac{1}{\sqrt{2}} \\
        \frac{1}{\sqrt{2}} & -\frac{1}{\sqrt{2}}
    \end{bmatrix}
    \begin{bmatrix}
        a + ib \\
        c + id
    \end{bmatrix} \\
    &=
    \begin{bmatrix}
    \frac{a}{\sqrt{2}} + \frac{c}{\sqrt{2}} + i\left(\frac{b}{\sqrt{2}} + \frac{d}{\sqrt{2}}\right) \\
    \frac{a}{\sqrt{2}} - \frac{c}{\sqrt{2}} + i\left(\frac{b}{\sqrt{2}} - \frac{d}{\sqrt{2}}\right) \\
    \end{bmatrix} \\
    &=
    \begin{bmatrix}
        (a_{1} + c_{1}) +  i(b_{1} + d_{1}) \\
        (a_{1} - c_{1}) +  i(b_{1} - d_{1})
    \end{bmatrix}
\end{split}
\end{equation*}

where $a_{1} = \frac{a}{\sqrt{2}}, b_{1} = \frac{b}{\sqrt{2}}, c_{1} =
\frac{c}{\sqrt{2}}, d_{1} = \frac{d}{\sqrt{2}}$.

\noindent
Hence, we can reduce the number of operations associated with matrix-vector multiplication by taking note of the
general structure of applying the Hadamard gate.

\subsubsection{Phase Gate}

\begin{equation*}
    P =
    \begin{bmatrix}
        1 &  0 \\
        0 & \cos(\phi) + i \sin(\phi)
    \end{bmatrix}
\end{equation*}

Let the two sides of the pair be $\gamma_{0}, \gamma_{1}$, where $\gamma_{0} = a + ib$ and $\gamma_{1} = c + id$.
The application of the $P$-gate to this pair can be expressed as follows:

\begin{equation*}
\begin{split}
    P\begin{bmatrix} \gamma_{0} \\ \gamma_{1} \end{bmatrix} &=
    \begin{bmatrix}
        1 &  0 \\
        0 & \cos(\phi) + i \sin(\phi)
    \end{bmatrix}
    \begin{bmatrix}
        \gamma_{0} \\
        \gamma_{1}
    \end{bmatrix} \\
    &=
    \begin{bmatrix}
        1 &  0 \\
        0 & \cos(\phi) + i \sin(\phi)
    \end{bmatrix}
    \begin{bmatrix}
        a + ib \\
        c + id
    \end{bmatrix} \\
    &=
    \begin{bmatrix}
         a + ib \\
        c \cos(\phi) - d \sin(\phi) + i\left(c \sin(\phi) + d \cos(\phi) \right)
    \end{bmatrix}
\end{split}
\end{equation*}

\noindent
The $Z$-gate is a special case of the $P$-gate, so the $P$-gate is able to leverage similar optimizations.
Namely, the structure of applying on the $P$-gate to a pair creates an opportunity to avoid the operations
associated with the general case of matrix-vector multiplication.
Moreover, the $P$-gate only requires one side of the pair to be modified.
Hence, the application of the $P$-gate allows for greater use of spatial locality.

\subsubsection{$R_{x}$-Gate}

The $R_{x}$-gate requires general matrix-vector multiplication, so we leverage fused multiply-add (FMA) instructions to
improve the performance of matrix-vector multiplication.

\subsubsection{$R_{y}$-Gate}

Similar to the $R_{x}$-gate, the $R_{y}$-gate requires general matrix-vector multiplication, so we leverage FMA
instructions to improve the performance of matrix-vector multiplication.

\subsubsection{$R_{z}$-Gate}

\begin{equation*}
    R_{z}(\lambda) =
    \begin{bmatrix}
        e^{-i \frac{\lambda}{2}} &  0 \\
        0                        & e^{i \frac{\lambda}{2}}
    \end{bmatrix}
\end{equation*}

Let the two sides of the pair be $\gamma_{0}, \gamma_{1}$, where $\gamma_{0} = a + ib$ and $\gamma_{1} = c + id$.
The application of the $R_{z}$-gate to this pair can be expressed as follows:

\begin{equation*}
\begin{split}
    R_{z}(\lambda) \begin{bmatrix} \gamma_{0} \\ \gamma_{1} \end{bmatrix}
    &=
    \begin{bmatrix}
        e^{-i \frac{\lambda}{2}} &  0 \\
        0                        & e^{i \frac{\lambda}{2}}
    \end{bmatrix}
    \begin{bmatrix}
        \gamma_{0} \\
        \gamma_{1}
    \end{bmatrix} \\
    &=
    \begin{bmatrix}
        e^{-i \frac{\lambda}{2}} &  0 \\
        0                        & e^{i \frac{\lambda}{2}}
    \end{bmatrix}
    \begin{bmatrix}
        a + ib \\
        c + id
    \end{bmatrix} \\
    &=
    \begin{bsmallmatrix}
         a \cos\left(\frac{\lambda}{2}\right) + b \sin\left(\frac{\lambda}{2}\right) + i\left(b \cos\left(\frac{\lambda}{2}\right) - a \sin\left(\frac{\lambda}{2}\right) \right) \\
         c \cos\left(\frac{\lambda}{2}\right) - d \sin\left(\frac{\lambda}{2}\right) + i\left(c \sin\left(\frac{\lambda}{2}\right) + d \sin\left(\frac{\lambda}{2}\right) \right)
    \end{bsmallmatrix}
\end{split}
\end{equation*}

Since the $R_{z}$ matrix is sparse, we can reduce the total number of operations associated with matrix-vector multiplication.

\subsubsection{U-Gate}

Since the $U$-gate is the most general single-qubit gate, we leverage FMA instructions to improve the performance of
matrix-vector multiplication.

\subsection{\label{sec:simd}SIMD}

SIMD (Single Instruction Multiple Data) capable CPUs have been available since
the introduction of the MMX instruction set~\cite{intel}. In essence, SIMD
operations enable the processing of multiple data with a single instruction.
At compile time, Spinoza discerns the CPU type of the local machine, which
allows for the enablement of all instruction subsets supported by the local
machine. Modern compilers are
able to generate highly optimized series of instructions when provided with
simple and predictable code. As such, the loops containing the ``fast paths''
in the strategies for each gate are manipulated such that the
generated code is as long as possible, while only containing low-latency
instructions such as \texttt{add}, \texttt{mul}, etc.

Spinoza's implementation precludes the need for manual SIMD\@. The compiler is
able to optimize the code given the highly optimized ``fast paths'', coupled
with the memory layout. The lack of manual SIMD increases the portability of
Spinoza. Although manual SIMD can potentially reduce the number of generated,
arithmetic, instructions, it is important to note that quantum state simulation
is highly memory bound. As a result, the potential savings from handwritten
SIMD would not provide any tangible performance benefits. Moreover, use of
manual SIMD would necessarily increase the size and complexity of Spinoza due
to the need to account for single-precision numbers, double precision numbers,
as well as a myriad of architectures.

\section{\label{sec:benchmarks}Benchmarks \& Comparisons}

We compare the performance of Spinoza against several publicly available quantum simulators.
We used the public GitHub repository~\cite{benchmark_repo}, in which benchmarking code has been reviewed by
contributors from various libraries.
We use the environment setup and configurations as provided in the repository, and compare against the equivalent
settings for Spinoza.
We used the latest version of each library, as listed in Table~\ref{tab:1}, unless otherwise specified in the
environment setup files provided in the benchmarking repository (i.e. Qiskit).
We use Spinoza's python bindings for a relevant comparison.
The implementation of the parameterized quantum circuit used for benchmarking can be found in the
repository~\cite{benchmark_repo}.

%\begin{table}[!htbp]
\begin{table}[!h]
    \centering
    \begin{tabular}{ l|l|l  }
        \hline
        Library & Language & Version \\
        \hline
        Qulacs~\cite{qulacs} & C++/Python & 0.6.0 \\
        ProjectQ~\cite{project} & C++/Python & 0.8.0 \\
        Qiskit~\cite{qiskit} & C++/Python & 0.16.0 \\
        Pennylane~\cite{pennylane} & Python & 0.30.0 \\
        Spinoza & Rust/Python & 0.1.0 \\
        \hline
    \end{tabular}
    \caption{Packages used in benchmarking. }
    \label{tab:1}
\end{table}

\begin{figure}[!htbp]
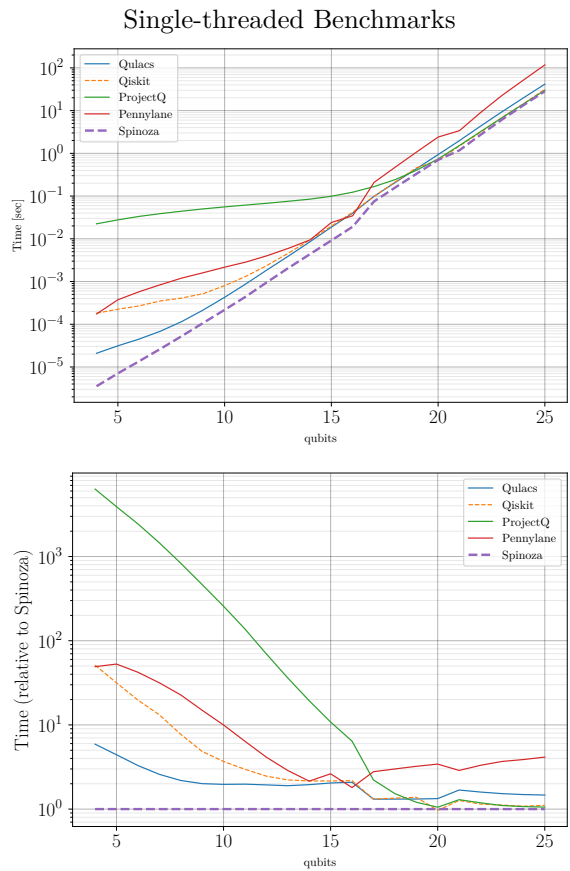

    \centering
    \begin{minipage}{.5\textwidth}
        \centering
        \title{Single-threaded Benchmarks}
        \scalebox{0.45}{\input{images/fig_compare_singlethreaded.pgf}}
    \end{minipage}
    \bigskip
    \begin{minipage}{.5\textwidth}
        \centering
        \scalebox{0.45}{\input{images/fig_ratio_singlethreaded.pgf}}
    \end{minipage}
    \caption{Benchmark times for simulating random quantum circuits, as implemented in
    ~\cite{benchmark_repo}, using a single thread.}
    \label{fig:benchmarks}
\end{figure}

All benchmarks were performed on a \texttt{c2-standard-16} instance, with 1 vCPU per core ratio, on Google Cloud
Platform (GCP).
Spinoza is compiled with \texttt{rustc} version $1.67.1$,and the following compiler flags: \texttt{-C opt-level=3 -C
target-cpu=native --edition=2021}.
In addition, we set \texttt{codegen-units = 1}, \texttt{lto = true}, \texttt{panic = "abort"}.
Lastly, all benchmarks were run using Python $3.8.16$.

Additional benchmarks using Spinoza in Rust and Qulacs in C++ are in Appendix~\ref{sec:rust_benchmarks}.

\subsection{Conclusion and Future Work}

Currently, Spinoza is capable of fast quantum state simulation on a myriad of machines when it is only given one
thread of execution to leverage.
Moreover, Spinoza has an edge over other simulators as shown in Section~\ref{sec:benchmarks}.
In future work, Spinoza will be parallelized.
In addition, Spinoza will be updated to allow for distributed quantum state simulation.

%%%%%%%%%%%%%%%%%%%%%%%%%%%%%%%%%%%%%%%%%%%%%%%%%%%%%%%%%%%%%%%%%%%%%%%%%%%
\acknowledgements
The authors thank Orson R. L. Peters for discussions on code optimization and
review. The authors thank Vitaliy Dorum for discussions on design and overall
code review.

The views expressed in this article are those of the authors and do not
represent the views of Wells Fargo. This article is for informational purposes only.
Nothing contained in this article should be construed as investment advice. Wells Fargo
makes no express or implied warranties and expressly disclaims all legal, tax,
and accounting implications related to this article.\\
%%%%%%%%%%%%%%%%%%%%%%%%%%%%%%%%%%%%%%%%%%%%%%%%%%%%%%%%%%%%%%%%%%%%%%%%%%%

%%%%%%%%%%%%%%%%%%%%%%%%%%%%%%%%%%%%%%%%%%%%%%%%%%%%%%%%%%%%%%%%%%%%%%%%%%%
\bibliographystyle{unsrtnat}
\bibliography{refs} % Entries are in the refs.bib file
%%%%%%%%%%%%%%%%%%%%%%%%%%%%%%%%%%%%%%%%%%%%%%%%%%%%%%%%%%%%%%%%%%%%%%%%%%%

%%%%%%%%%%%%%%%%%%%%%%%%%%%%%%%%%%%%%%%%%%%%%%%%%%%%%%%%%%%%%%%%%%%%%%%%%%%
\onecolumn
\renewcommand\floatpagefraction{0.9}
\appendix
%%%%%%%%%%%%%%%%%%%%%%%%%%%%%%%%%%%%%%%%%%%%%%%%%%%%%%%%%%%%%%%%%%%%%%%%%%%

\section{\label{sec:rust_code}Data Structures and Pair Strategies in Rust}

\begin{lstlisting}[caption={State structure defintion. \texttt{Float} can be type \texttt{f32} or \texttt{f64}---depending on the feature flag chosen.},
    label={statevector}]
pub struct State {
    pub reals: Vec<Float>,
    pub imags: Vec<Float>,
    pub n: u8,
}
\end{lstlisting}

\begin{lstlisting}[caption={Rust implementation Strategy 0 - Traverse and Recognize.},label={strat0}]
for i in 0..(1 << state.n) {
    let k = (i >> target) & 1;
    let m = &self.m[k];
    // UPDATE AMPLITUDES
}
\end{lstlisting}

\begin{lstlisting}[caption={Rust implementation Strategy 1 - Group and Traverse.},
    label={strat1}]
let mut chunk_start = range.start;

while chunk_start < range.end {
    let chunk_end = range.end.min(((chunk_start >> target) + 1) << target);

    let m = self.matrix[(chunk_start >> target) & 1];
    for i in chunk_start..chunk_end {
        // UPDATE AMPLITUDES
    }
    chunk_start = chunk_end;
}
\end{lstlisting}

\begin{lstlisting}[caption={Rust implementation of Strategy 2 - Concatenation.},label={strat2}]
let end = state.imags.len() >> 1;
let chunks = end >> target;

(0..chunks).into_iter().for_each(|chunk| {
    let dist = 1 << target;
    let base = (2 * chunk) << target;
    for i in 0..dist {
        let l0 = base + i;
        let l1 = l0 + dist;
        // UPDATE AMPLITUDES
    }
})
\end{lstlisting}

\begin{lstlisting}[caption={Rust implementation of Strategy 3 - Insertion.},label={strat3}]
let dist = 1 << target;
let num_pairs = state.len() >> 1;

for i in 0..num_pairs {
    let l0 = i + ((i >> target) << target);
    let l1 = l0 + dist;
    // UPDATE AMPLITUDES
}
\end{lstlisting}

\begin{lstlisting}[caption={Rust implementation of Strategy 3 - Insertion, for controlled gate transformations.},
    label={strat3_control}]
let end = state.imags.len() >> 2;

let dist = 1 << target;
let marks = (target.min(control), target.max(control));

for i in 0..end {
    let x = i + (1 << (marks.1 - 1)) + ((i >> (marks.1 - 1)) << (marks.1 - 1));
    let l1 = x + (1 << marks.0) + ((x >> marks.0) << marks.0);
    let l0 = l1 - dist;
    // UPDATE AMPLITUDES
}
\end{lstlisting}

\section{\label{sec:rust_examples}Using Spinoza: Rust}

\begin{lstlisting}[caption={Value Encoding, as described in~\cite{interpolation_paper}, implemented using Spinoza.},
    label={functional_example}]
fn value_encoding(n: usize, v: Float) {
    let mut state = State::new(n);

    for i in 0..n {
        apply(Gate::H, &mut state, i);
    }

    for i in 0..n {
        apply(Gate::P(2.0 * PI / (pow2f(i + 1)) * v), &mut state, i);
    }

    let targets: Vec<usize> = (0..n).rev().collect();
    iqft(&mut state, &targets);
}
\end{lstlisting}

\begin{lstlisting}[caption={Value Encoding circuit, as described in~\cite{interpolation_paper}, implemented using Spinoza.},
    label={circuit_example}]
fn value_encoding(n: usize, v: Float) {
    let mut q = QuantumRegister::new(n);
    let mut qc = QuantumCircuit::new(&mut q);

    for i in 0..n {
        qc.h(q[i])
    }

    for i in 0..n {
        qc.p((2 as Float) * PI / pow2f(i + 1) * v, q[i])
    }

    let targets: Vec<usize> = (0..n).rev().collect();
    qc.iqft(&targets);

    qc.execute();
}

\end{lstlisting}

\section{\label{sec:python_examples}Using Spinoza: Python}

\begin{lstlisting}[language=Python, caption={Value Encoding circuit, as described in~\cite{interpolation_paper},
implemented using the Spinoza Python Library.}, label={py_circuit_example}]
def value_encoding(n, v):
    q = QuantumRegister(n)
    qc = QuantumCircuit(q)

    for i in range(n):
        qc.h(i)

    for i in range(n):
        qc.p(2 * np.pi / (2 ** (i + 1)) * v, i)

    qc.iqft(range(n)[::-1])

    return qc
\end{lstlisting}

\vspace{1cm}
\twocolumn
\begin{lstlisting}[language=Python, caption={Using the resulting statevector from running the value encoding circuit
defined in listing~\ref{py_circuit_example} on three qubits with the parameter 2.4 we get 1,000 samples using
reservoir sampling}, label={py_sampling_example}]
circuit = value_encoding(3, 2.4)
state = run(circuit)

samples = get_samples(state, 1000)
\end{lstlisting}

\begin{figure}[!h]
    \centering
    \includegraphics[width=.95\linewidth]{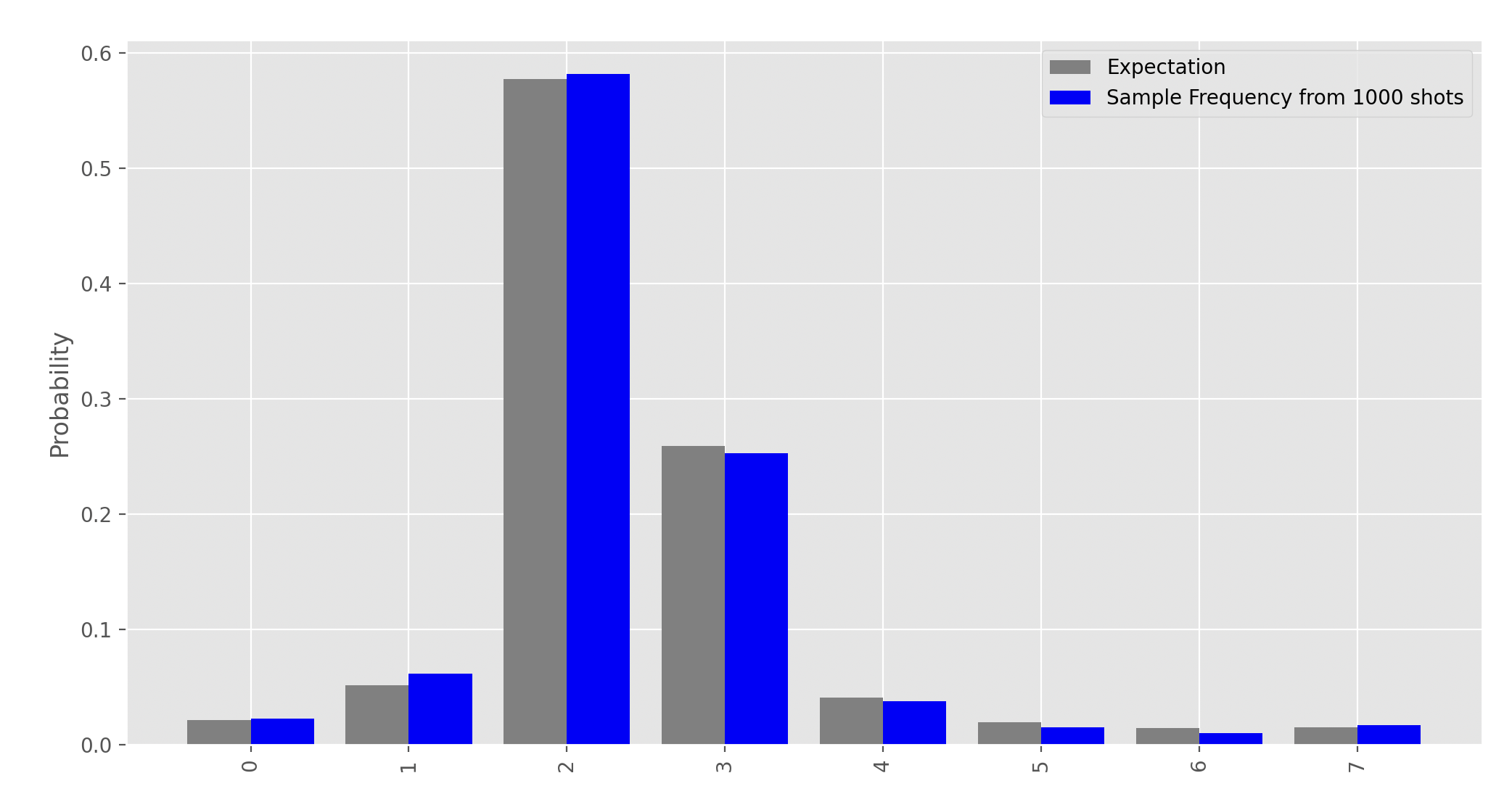}
    \captionof{figure}{Expected probabilities versus sample frequencies from the \texttt{state} and \texttt{samples}
    in listing~\ref{py_sampling_example}.}
\end{figure}

\newpage
\onecolumn
\section{\label{sec:rust_benchmarks}Rust Benchmarks}

The following benchmarks were implemented using the Rust implementation of Spinoza and the C++ implementation of Qulacs.
This allows for comparison without concern for the overhead from using the respective Python libraries.

\begin{figure}[!ht]
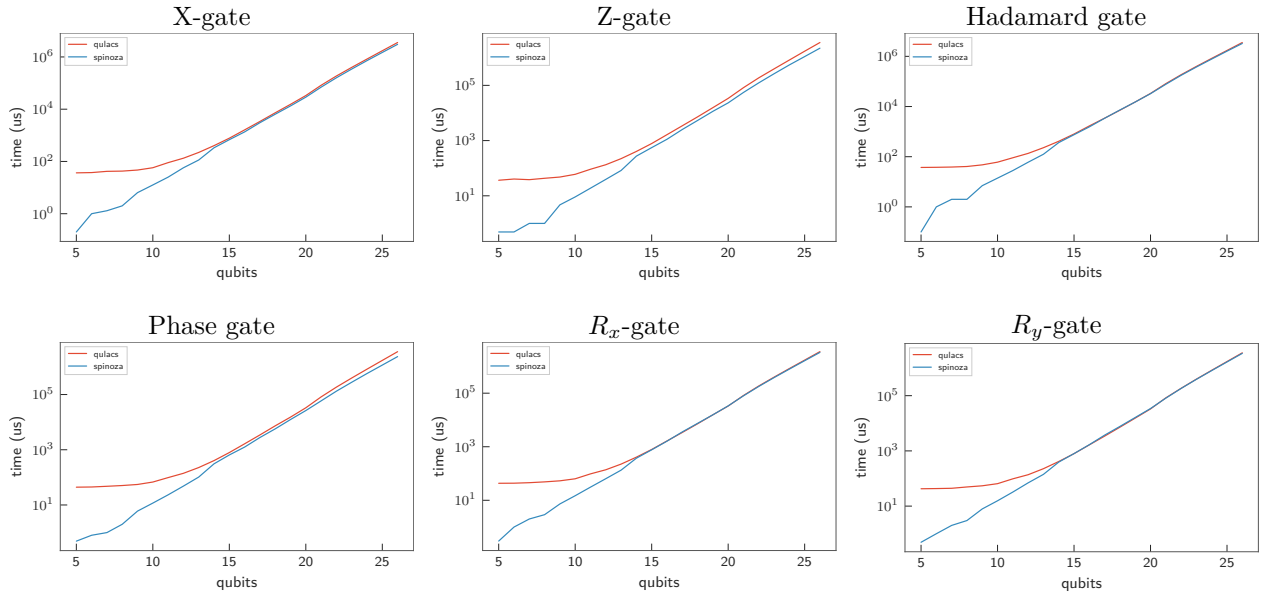

    \centering
    \begin{minipage}{.32\textwidth}
        \centering
        \title{X-gate}
        \scalebox{0.65}{\input{images/x-gate-log-linear.pgf}}
    \end{minipage}
    \begin{minipage}{.32\textwidth}
        \centering
        \title{Z-gate}
        \scalebox{0.65}{\input{images/z-gate-log-linear.pgf}}
    \end{minipage}
    \begin{minipage}{.32\textwidth}
        \centering
        \title{Hadamard gate}
        \scalebox{0.65}{\input{images/h-gate-log-linear.pgf}}
    \end{minipage}

    \bigskip
    \centering
    \begin{minipage}{.32\textwidth}
        \centering
        \title{Phase gate}
        \scalebox{0.65}{\input{images/p-gate-log-linear.pgf}}
    \end{minipage}
    \begin{minipage}{.32\textwidth}
        \centering
        \title{$R_{x}$-gate}
        \scalebox{0.65}{\input{images/rx-gate-log-linear.pgf}}
    \end{minipage}
    \begin{minipage}{.32\textwidth}
        \centering
        \title{$R_{y}$-gate}
        \scalebox{0.65}{\input{images/ry-gate-log-linear.pgf}}
    \end{minipage}
    \captionof{figure}{Average time (logarithmic scale) to execute ten iterations of applying the given gate to each qubit in a system
    for five to twenty-six qubits.}
    \label{fig:rust_benchmarks}
\end{figure}

\begin{figure}[!htbp]
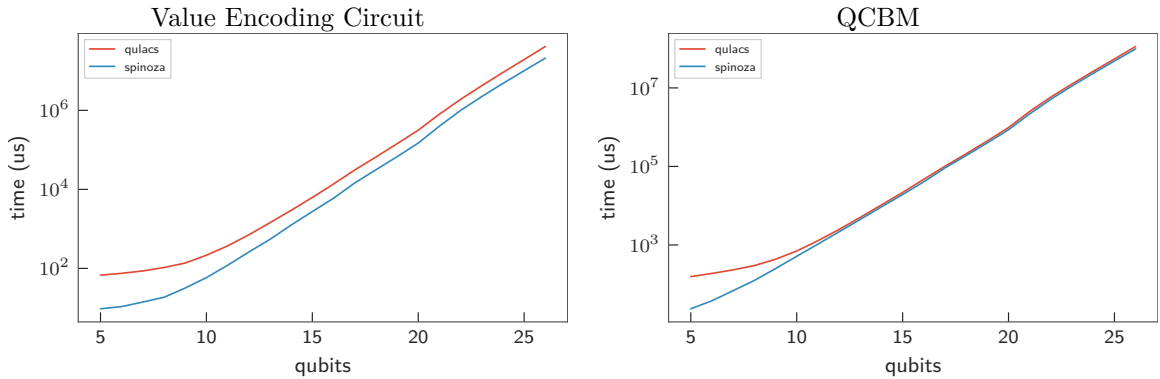

    \centering
    \begin{minipage}{.45\textwidth}
        \centering
        \title{Value Encoding Circuit}
        \scalebox{0.9}{\input{images/value_encoding-gate-log-linear.pgf}}
    \end{minipage}
    \begin{minipage}{.45\textwidth}
        \centering
        \title{QCBM}
        \scalebox{0.9}{\input{images/qcbm-gate-log-linear.pgf}}
    \end{minipage}
    \captionof{figure}{Average time (logarithmic scale) to execute ten iterations of the "QCBM" quantum circuits (as implemented in
    ~\cite{benchmark_repo}) for five to twenty-six qubits.}
\end{figure}

\end{document}